\documentclass[copyright,creativecommons]{eptcs}
\providecommand{\event}{AFL 2023} %

\usepackage{iftex}

\ifpdf
  \usepackage{underscore}         %
  \usepackage[T1]{fontenc}        %
\else
  \usepackage{breakurl}           %
\fi

\usepackage{ifthen}

\usepackage{cite}

\usepackage{hyperref}
\usepackage{nameref}

\usepackage{amsmath}
\usepackage{amsfonts}
\usepackage{amssymb}
\usepackage{amsthm}

\usepackage{cleveref}
\theoremstyle{plain}
\newtheorem{theorem}{Theorem}
\crefname{theorem}{Thm.}{Thms.}
\newtheorem{proposition}[theorem]{Proposition}
\crefname{proposition}{Prop.}{Props.}
\newtheorem{lemma}[theorem]{Lemma}
\crefname{lemma}{Lem.}{Lems.}
\newtheorem*{sublemma*}{Sublemma}
\newtheorem{corollary}[theorem]{Corollary}
\crefname{corollary}{Cor.}{Cors.}

\theoremstyle{definition}
\newtheorem{definition}[theorem]{Definition}
\crefname{definition}{Def.}{Defs.}

\theoremstyle{remark}

\newtheorem{remark}[theorem]{Remark}

\crefname{figure}{Fig.}{Figs.}

\crefname{section}{Sect.}{Sects.}
\crefname{appendix}{Appendix}{}

\usepackage{tikz}
\usetikzlibrary{shapes, arrows, arrows.spaced, arrows.meta, chains,positioning, cd,
    decorations, decorations.pathmorphing, decorations.markings, decorations.pathreplacing, automata, backgrounds, petri, matrix, fit, calc, graphs, quotes, bending}

\usepackage{ebproof}
\ebproofnewstyle{small}{separation = .5em, label separation= .1em, right label template= \footnotesize \inserttext}
\usepackage{stackengine}

\usepackage{empheq}

\usepackage{textcomp}
\usepackage{stmaryrd} %
\usepackage{stackengine} %

\usepackage{tablefootnote}
\usepackage{multirow}
\usepackage{diagbox}

\usepackage{algorithm, algpseudocode}

\usepackage{pict2e,picture}
\usepackage{xparse}
\usepackage{accents}
\usepackage{xcolor}
\definecolor[named]{ACMBlue}{cmyk}{1,0.1,0,0.1}
\definecolor[named]{ACMYellow}{cmyk}{0,0.16,1,0}
\definecolor[named]{ACMOrange}{cmyk}{0,0.42,1,0.01}
\definecolor[named]{ACMRed}{cmyk}{0,0.90,0.86,0}
\definecolor[named]{ACMLightBlue}{cmyk}{0.49,0.01,0,0}
\definecolor[named]{ACMGreen}{cmyk}{0.20,0,1,0.19}
\definecolor[named]{ACMPurple}{cmyk}{0.55,1,0,0.15}
\definecolor[named]{ACMDarkBlue}{cmyk}{1,0.58,0,0.21}
\hypersetup{colorlinks,
    linkcolor=ACMPurple,
    citecolor=ACMPurple,
    urlcolor=ACMDarkBlue,
    filecolor=ACMDarkBlue}

\usepackage[normalem]{ulem}

\usepackage{enumitem}
\makeatletter
\let\orgdescriptionlabel\descriptionlabel
\renewcommand*{\descriptionlabel}[1]{%
    \let\orglabel\label
    \let\label\@gobble
    \phantomsection
    \edef\@currentlabel{#1\unskip}%
    \let\label\orglabel
    \orgdescriptionlabel{#1}%
}
\makeatother

\usepackage{etoolbox}%
\patchcmd{\footnotemark}{\stepcounter{footnote}}{\refstepcounter{footnote}}{}{}

\usepackage[notion]{knowledge}

\usepackage{snotez} %
\newcommand{\set}[1]{\{ #1 \}}
\newcommand{\tuple}[1]{\langle #1 \rangle}
\newcommand{\const}[1]{\mathsf{#1}}
\newcommand{\bl}{\_}
\newcommand{\defeq}{\mathrel{\ensurestackMath{\stackon[1pt]{=}{\scriptscriptstyle\Delta}}}}
\newcommand{\defiff}{\mathrel{\ensurestackMath{\stackon[1pt]{\iff}{\scriptscriptstyle\Delta}}}}

\newcommand{\nat}{\mathbb{N}}
\newcommand{\card}{\mathop{\#}}

\NewDocumentCommand\word{O{1}}{%
    \ifcase#1
        undefined
    \or w
    \or v
    \else undefined

    \fi
}
\newcommand{\len}[1]{\|#1\|}
\NewDocumentCommand\la{O{1}}{%
    \ifcase#1
        undefined
    \or L
    \or K
    \else undefined

    \fi
}

\newcommand{\vsig}{\mathbf{V}}

\newcommand{\Term}{\mathbf{T}}
\NewDocumentCommand\term{O{1}}{%
    \ifcase#1
        undefined
    \or t
    \or s
    \or u
    \else undefined

    \fi
}

\newcommand{\sig}{S}
\newcommand{\union}{\mathbin{\cup}}
\newcommand{\compo}{\mathbin{\cdot}}
\newcommand{\kstar}{*}
\newcommand{\emp}{\bot}
\newcommand{\eps}{\const{I}}
\newcommand{\compl}{-}

\newcommand{\domain}[1]{|#1|}
\newcommand{\val}{\mathfrak{v}}
\newcommand{\lang}[1]{[ #1 ]}
\newcommand{\LANG}{\mathsf{LANG}}

\newcommand{\LENGTH}[1]{\mathsf{L}_{#1}}

\knowledge{ignore}
| letter
| letters

\knowledge{ignore}
| word
| words

\knowledge{ignore}
| length

\knowledge{ignore}
| language
| languages

\knowledge{notion}
| language of a term

\knowledge{ignore}
| composition

\knowledge{ignore}
| Kleene star

\knowledge{notion}
| $\sig$-algebra
| $\sig$-algebras

\knowledge{notion}
| equation
| equations

\knowledge{notion}
| inequation
| inequations

\knowledge{notion}
| equational theory w.r.t.\ languages

\knowledge{notion}
| valuation
| valuations

\knowledge{notion}
| language model
| language models

\knowledge{notion}
| language valuation
| language valuations
| (language) valuations

\knowledge{notion}
| composition-free

\knowledge{notion}
| star-free

\knowledge{notion}
| words-to-letters valuation
| words-to-letters valuations
| Words-to-letters valuations

\knowledge{ignore}
| variable
| variables

\knowledge{notion}
| KA term
| KA terms
 
\usepackage{environ}
\newboolean{isdraft}

\setboolean{isdraft}{true}

\ifthenelse{\boolean{isdraft}}{
  \setlength{\marginparwidth}{7em}
  \setlength{\marginparsep}{0.1em}
  \NewEnviron{commentyn}{\sidenote{\textcolor{green}{YN: \BODY}}}
}{
  \NewEnviron{commentyn}{}
}

\title{Words-to-Letters Valuations for Language Kleene Algebras with Variable Complements}
\author{Yoshiki Nakamura
\institute{Tokyo Institute of Technology, Japan}
\email{nakamura.yoshiki.ny@gmail.com}
\and
Ryoma Sin'ya
\institute{Akita University, Japan}
\email{ryoma@math.akita-u.ac.jp}
}
\def\titlerunning{Words-to-Letters Valuations for Language Kleene Algebras with Variable Complements}
\def\authorrunning{Y. Nakamura \& R. Sin'ya}
\begin{document}
\maketitle

\begin{abstract}
We investigate the equational theory of Kleene algebra terms with \emph{variable complements}---(language) complement where it applies only to variables---w.r.t.\ languages.
While the equational theory w.r.t.\ languages coincides with the language equivalence (under the standard language valuation) for Kleene algebra terms, this coincidence is broken if we extend the terms with complements.
In this paper, we prove the decidability of some fragments of the equational theory: the universality problem is coNP-complete, and the inequational theory $t \le s$ is coNP-complete when $t$ does not contain Kleene-star.
To this end, we introduce \emph{words-to-letters valuations};
they are sufficient valuations for the equational theory and ease us in investigating the equational theory w.r.t.\ languages.
Additionally, we prove that for words with variable complements, the equational theory coincides with the word equivalence.
 \end{abstract}

\section{Introduction}
Kleene algebra (KA) \cite{kleeneRepresentationEventsNerve1951, Conway1971} is an algebraic system for regular expressions consisting of union ($\union$), composition ($\compo$), Kleene-star ($\bl^{\kstar}$), emptiness ($\emp$), and identity ($\eps$).
In this paper, we consider KAs \emph{w.r.t.\ languages} (a.k.a., \kl{language models} of KAs, language KAs).
Interestingly, the equational theory of KAs w.r.t.\ languages coincides with the language equivalence under the standard language valuation (see also, e.g., \cite{andrekaEquationalTheoryKleene2011,pousCompletenessTheoremsKleene2022}):
for all \kl{KA terms} (i.e., regular expressions) $\term[1], \term[2]$, we have
\begin{align*}
  \label{equation: LANG = lang}\LANG \models \term[1] = \term[2] \quad \iff \quad \lang{\term[1]} = \lang{\term[2]} \tag{$\dagger$}.
\end{align*}
Here, we write $\LANG \models \term[1] = \term[2]$ if the equation $\term[1] = \term[2]$ holds for all \kl{language models} (i.e., each variable $x$ maps to not only the singleton language $\set{x}$ but also any languages);
we write $\lang{\term[3]}$ for the language of a regular expression $\term[3]$ (i.e., each variable $x$ maps to the singleton language $\set{x}$).
Since the valuation $\lang{\bl}$ is an instance of valuations in $\LANG$, the direction $\Longrightarrow$ is trivial (this direction always holds even if we extend KA terms with some extra operators).
The direction $\Longleftarrow$ is a consequence of the completeness of KAs (see \Cref{section: LANG and lang} for an alternative proof not relying on the completeness of KAs).
However, the direction $\Longleftarrow$ fails in general when we extend KA terms with extra operators.
Namely, the equational theory w.r.t.\ languages does not coincide with the language equivalence in general (see below for complements).
The equational theory of KAs with some operators w.r.t.\ languages was studied,
e.g., with reverse \cite{bloomNotesEquationalTheories1995},
with tests \cite{kozenKleeneAlgebraTests1996} (where languages are of guarded strings, not words),
with intersection \cite{andrekaEquationalTheoryKleene2011},
and with universality ($\top$) \cite{pousCompletenessTheoremsKleene2022}.
Nevertheless, to the best of authors' knowledge, \emph{variable complements} (and even complements) w.r.t. languages has not yet been investigated,
while those w.r.t.\ binary relations were studied, e.g.,\ in \cite{ngRelationAlgebrasTransitive1984} (for complements; cf.\ Tarski's calculus of relations \cite{Tarski1941}) and \cite{nakamuraExistentialCalculiRelations2023} (for variable complements).

In this paper, we investigate the equational theory of KA terms with \emph{variable complements} ($x^{\compl}$)---(language) complement, where it applies only to variables (we use $x$ to denote variables)---w.r.t.\ languages.
For KA terms with variable complements, (\ref{equation: LANG = lang}) fails.
The following is an example:
\begin{align*}
  \LANG & \not\models x^{\compl} = x^{\compl} \compo x^{\compl} & \lang{x^{\compl}} & = \lang{x^{\compl} \compo x^{\compl}}.
\end{align*}
(For $\LANG \not\models$, consider a valuation such that $x^{\compl}$ maps to the language $\set{a}$.)
As the example above (see also \Cref{remark: LANG and lang}, for more examples) shows, the equational theory of KAs with variable complements w.r.t.\ languages significantly differs from the language equivalence under the standard language valuation.
While the language equivalence problem of KA terms with variable complements is decidable in PSPACE by a standard automata construction \cite{thompsonProgrammingTechniquesRegular1968} (and hence, PSPACE-complete \cite{meyerEquivalenceProblemRegular1972, Meyer1973, Hunt1976}),
it remains whether the equational theory w.r.t.\ languages is decidable.

We prove the decidability and complexity of some fragments of the equational theory of KA terms with variable complements w.r.t.\ languages:
the universality problem is coNP-complete (\Cref{corollary: universality decidable}), and the inequational theory $\term[1] \le \term[2]$ is coNP-complete when $\term[1]$ does not contain Kleene-star (\Cref{corollary: star-free}).
To this end, we introduce \emph{\kl{words-to-letters valuations}}.
\kl{Words-to-letters valuations} are sufficient for the equational theory of KA terms with variable complements w.r.t.\ languages (words-to-letters-valuation property; \Cref{corollary: word witness}):
for all terms $\term[1], \term[2]$, if there is some \kl{language valuation} such that it refutes $\term[1] = \term[2]$, there is some \kl{words-to-letters valuation} such that it refutes $\term[1] = \term[2]$.
This property eases us in investigating the equational theory w.r.t.\ languages.

Additionally, by using \kl{words-to-letters valuations}, we prove a completeness theorem for words with variable complements: the equational theory coincides with the word equivalence (\Cref{theorem: completeness word}).
A limitation of \kl{words-to-letters valuations} is that the number of \kl{letters} is not bounded; so, they cannot apply to the equational theory over $\LANG_{n}$ (language models over sets of cardinality at most a natural number $n$).
Nevertheless, by giving another valuation, we can show the coincidence for one-variable words (\Cref{theorem: completeness word one variable}).
We leave open for the many-variable words.

\subsection*{Outline}
In \Cref{section: preliminaries}, we briefly give basic definitions, including the syntax and semantics of KA terms with variable complements.
Additionally, we give languages for KA terms with variable complements (\Cref{section: term to lang}).
In \Cref{section: identity,section: univ,section: word}, we consider fragments of the equational theory of KA terms with variable complements w.r.t.\ languages, step-by-step.
In \Cref{section: identity}, we consider the identity inclusion problem ($\LANG \models \eps \le \term$?).
This problem is relatively easy but contains the coNP-hardness result (\Cref{corollary: identity decidable}).
In \Cref{section: univ}, we consider the variable inclusion problem ($\LANG \models x \le \term$?) and the universality problem ($\LANG \models \top \le \term$?).
In \Cref{section: word}, we consider the word inclusion problem ($\LANG \models \word \le \term$?).
This section proceeds in the same way as \Cref{section: univ}, thanks to \kl{words-to-letters valuations} (\Cref{definition: val for words to letters}).
Consequently, the inequational theory $\term[1] \le \term[2]$ is coNP-complete when $\term[1]$ does not contain Kleene-star (\Cref{corollary: star-free}).
Additionally, we prove the words-to-letters valuation property (\Cref{corollary: word witness}) for the equational theory of (full) KA terms with variable complements w.r.t.\ languages.
In \Cref{section: completeness word}, we consider the equational theory of \kl{words} with variable complements and show a completeness theorem (\Cref{theorem: completeness word}).
\Cref{section: conclusion} concludes this paper.

\section{Preliminaries}\label{section: preliminaries}
We write $\nat$ for the set of non-negative integers.
For a set $X$, we write $\card(X)$ for the cardinality of $X$ and $\wp(X)$ for the power set of $X$.

For a set $X$ (of \intro*\kl{letters}), we write $X^{\kstar}$ for the set of \intro*\kl{words} over $X$: finite sequences of elements of $X$.
We write $\eps$ for the empty word.
We write $\word[1] \word[2]$ for the concatenation of \kl{words} $\word[1]$ and $\word[2]$.
A \intro*\kl{language} over $X$ is a subset of $X^{\kstar}$.
We use $\word[1], \word[2]$ to denote \kl{words}
and use $\la[1], \la[2]$ to denote \kl{languages}, respectively.
For \kl{languages} $\la[1], \la[2] \subseteq X^{\kstar}$, the \intro*\kl{composition} $\la[1] \compo \la[2]$ and the \intro*\kl{Kleene star} $\la[1]^{\kstar}$ is defined by:
\begin{align*}
    \la[1] \compo \la[2] & \quad\defeq\quad \set{\word[1] \word[2] \mid \word[1] \in \la[1] \ \land \ \word[1] \in \la[2]};                        \\
    \la[1]^{\kstar}      & \quad\defeq\quad \set{\word[1]_0 \dots \word[1]_{n-1} \mid \exists n \in \nat, \forall i < n,\  \word[1]_i \in \la[1]}.
\end{align*}

\subsection{Syntax: KA terms with variable complements}
Let $\vsig$ be a set of \intro*\kl{variables}.
The set of Kleene algebra (KA) terms with variable complements ($x^{\compl}$) is defined by the following grammar:
\begin{align*}
    \Term \ni \term[1], \term[2], \term[3] \quad\Coloneqq\quad x \mid \eps \mid \emp \mid \term[1] \compo \term[2]  \mid \term[1] \union \term[2]\mid \term[1]^{\kstar} \mid x^{\compl} \tag{$x \in \vsig$}
\end{align*}
We use parentheses in ambiguous situations.
We often abbreviate $\term[1] \compo \term[2]$ to $\term[1] \term[2]$.
We write $\top$ for the term $x \union x^{\compl}$, where $x$ is any \kl{variable}.

An \intro*\kl{equation} $\term[1] = \term[2]$ is a pair of terms.
An \intro*\kl{inequation} $\term[1] \le \term[2]$ is an abbreviation of the \kl{equation} $\term[1] \union \term[2] = \term[2]$.
\subsection{Semantics: language models}
Consider the signature $\sig \defeq \set{\eps_{(0)}, \emp_{(0)}, \compo_{(2)}, \union_{(2)}, {\bl^{\kstar}}_{(1)}, {\bl^{\compl}}_{(1)}}$.
An \intro*\kl{$\sig$-algebra} $A$ is a tuple $\tuple{\domain{A}, \set{f^{A}}_{f_{(k)} \in \sig}}$, where $\domain{A}$ is a non-empty set and $f^{A} \colon \domain{A}^{k} \to \domain{A}$ is a $k$-ary map for each $f_{(k)} \in \sig$.
A \intro*\kl{valuation} $\val$ on an \kl{$\sig$-algebra} $A$ is a map $\val \colon \vsig \to \domain{A}$.
For a \kl{valuation} $\val$, we write $\hat{\val} \colon \Term \to \domain{A}$ for the unique homomorphism extending $\val$.

The \intro*\kl{language model} $A$ over a set $X$ is an \kl{$\sig$-algebra} such that $\domain{A} = \wp(X^{\kstar})$ and
for all $\la[1], \la[2] \subseteq X^{\kstar}$,
\begin{flalign*}
    \phantom{x^{A}} & \phantom{= \set{x}} & \eps^{A}                 & = \set{\eps}                   &
                    &                     & \la[1] \compo^{A} \la[2] & = \la[1] \compo \la[2]         &
                    &                     & \la[1]^{\kstar^{A}}      & = \la[1]^{\kstar}                \\
                    &                     & \emp^{A}                 & = \emptyset                    &
                    &                     & \la[1] \union^{A} \la[2] & = \la[1] \union \la[2]         &
                    &                     & \la[1]^{\compl^{A}}      & = X^{\kstar} \setminus \la[1].
\end{flalign*}
We write $\LANG$ for the class of all \kl{language models}.
A \intro*\kl{language valuation} over a set $X$ is a \kl{valuation} on some language model over $X$.
For an \kl{equation} $\term[1] = \term[2]$, we let
\[\LANG \models \term[1] = \term[2] \quad \defiff \quad \hat{\val}(\term[1]) = \hat{\val}(\term[2]) \mbox{ holds for all \kl{language valuations} $\val$}.\]
The \intro*\kl{equational theory w.r.t.\ languages} is the set of all \kl{equations} $\term[1] = \term[2]$ such that $\LANG \models \term[1] = \term[2]$.

Additionally, the \intro*\kl[language of a term]{language} $\lang{\term} \subseteq \vsig^{\kstar}$ of a term $\term$ is defined by:
\begin{flalign*}
    \lang{x} & \defeq \set{x} & \lang{\eps}                     & \defeq \set{\eps}                                &
             &                & \lang{\term[1] \compo \term[2]} & \defeq \lang{\term[1]} \compo \lang{\term[2]}    &
             &                & \lang{\term[1]^{\kstar}}        & \defeq \lang{\term[1]}^{\kstar}                    \\
             &                & \lang{\emp}                     & \defeq \emptyset                                 &
             &                & \lang{\term[1] \union \term[2]} & \defeq \lang{\term[1]} \union \lang{\term[2]}    &
             &                & \lang{\term[1]^{\compl}}        & \defeq \vsig^{\kstar} \setminus \lang{\term[1]}.
\end{flalign*}
By definition, we have $\lang{\term} = \hat{\val}(\term)$ if $\val$ is the \kl{valuation} on the \kl{language model} over the set $\vsig$ defined by $\val(x) = \set{x}$ for all $x \in \vsig$.
Hence, for all $\term[1], \term[2]$, we have
\begin{align*}
    \label{equation: LANG to lang} \LANG \models \term[1] = \term[2] \quad \Longrightarrow \quad \lang{\term[1]} = \lang{\term[2]}. \tag{$\ddagger$}
\end{align*}
\begin{remark}\label{remark: LANG and lang}
    The converse direction fails\footnote{The failure can be also shown by that the universality $\top$ can be expressed by $x \union x^{\compl}$; see also \cite[Remark.\ 3.6]{pousCompletenessTheoremsKleene2022}.}; for example, when $x \neq y$,
    \begin{align*}
        \LANG \not\models y \le x^{\compl} &  & \lang{y} \subseteq \lang{x^{\compl}}.
    \end{align*}
    Here $\term[1] \le \term[2]$ denotes the \kl{equation} $\term[1] \union \term[2] = \term[2]$ (so, indeed, an \kl{equation}).
    For $\LANG \not\models y \le x^{\compl}$, consider a \kl{language valuation} $\val$ such that $a \in \val(x)$ and $a \in \val(y)$; then we have $a \in \hat{\val}(y) \setminus \hat{\val}(x^{\compl})$.
    $\lang{y} \subseteq \lang{x^{\compl}}$ is shown by $\lang{y} = \set{y} \subseteq \vsig^{\kstar} \setminus \set{x} = \lang{x^{\compl}}$.
    More generally, for any \kl{word} $\word$ over $\vsig$ such that $\word \neq x$,
    \begin{align*}
        \LANG \not\models \word \le x^{\compl} &  & \lang{\word} \subseteq \lang{x^{\compl}}.
    \end{align*}
    Moreover, for example, there are the following examples
    (for $\LANG \not\models$, consider a valuation such that both $x$ and $y$ map to the language $X^* \setminus \set{a}$, where $X$ is a set and $a \in X$):
    \begin{align*}
        \LANG & \not\models x^{\compl} = x^{\compl} \compo x^{\compl} & \lang{x^{\compl}} & = \lang{x^{\compl} \compo x^{\compl}}  \\
        \LANG & \not\models \top = x^{\compl} \compo y^{\compl}       & \lang{\top}       & = \lang{x^{\compl} \compo y^{\compl}}  \\
        \LANG & \not\models \top = x^{\compl} \union y^{\compl}       & \lang{\top}       & = \lang{x^{\compl} \union y^{\compl}}.
    \end{align*}

    As the examples above show, for KA terms with variable complements, the \kl{equational theory w.r.t.\ languages} ($\LANG \models \term[1] = \term[2]$?) significantly differs from the language equivalence problem ($\lang{\term[1]} = \lang{\term[2]}$?).
\end{remark}
In the sequel, we focus on the \kl{equational theory w.r.t.\ languages} and investigate its fragments.
We prepare a useful tool (\Cref{lemma: lang val}), which enables us to decompose terms into languages of words.

\subsection{Languages for KA terms with variable complements}\label{section: term to lang}
Let $\vsig' = \set{x, x^{\compl} \mid x \in \vsig}$.
For a term $\term$, we write $\lang{\term}_{\vsig'}$ for the \kl[language of a term]{language} of $\term$ where $\term$ is viewed as the regular expression over $\vsig'$.
Each \kl{word} over $\vsig'$ is a term such that both the union ($\union$) and the Kleene-star ($\bl^{\kstar}$) do not occur.
Note that $\lang{x^{\compl}}_{\vsig'} = \set{x^{\compl}}$, cf.\ $\lang{x^{\compl}} = \vsig^* \setminus \set{x}$.
For a \kl{language valuation} $\val$ and a \kl{language} $\la$ over $\vsig'$, we define
\[\hat{\val}(\la) \quad\defeq\quad \bigcup_{\word \in \la} \hat{\val}(\word). \]
By using the distributive law of $\compo$ w.r.t.\ $\union$, for all \kl{languages} $\la[1], \la[2]$ and \kl{language valuations} $\val$,
we have:
\begin{align*}
    \hat{\val}(\la[1] \union \la[2]) & = \hat{\val}(\la[1]) \union \hat{\val}(\la[2]) &
    \hat{\val}(\la[1] \compo \la[2]) & = \hat{\val}(\la[1]) \compo \hat{\val}(\la[2]) &
    \hat{\val}(\la[1]^{\kstar})      & = \hat{\val}(\la[1])^{\kstar}.
\end{align*}
We can decompose each term $\term$ to the set $\lang{\term}_{\vsig'}$ of \kl{words} over $\vsig'$ as follows:
\begin{lemma}\label{lemma: lang val}
    Let $\val$ be a \kl{language valuation}.
    For all terms $\term$, we have
    \[\hat{\val}(\term) \quad=\quad \hat{\val}(\lang{\term}_{\vsig'}).\]
\end{lemma}
\begin{proof}
    By easy induction on $\term$ using the equations above.
    Case $\term = x, x^{\compl}, \eps$:
    Clear, by $\lang{\term}_{\vsig'} = \set{\term}$.
    Case $\term = \bot$:
    By $\hat{\val}(\bot) = \emptyset = \hat{\val}(\lang{\bot}_{\vsig'})$.
    Case $\term = \term[2] \union \term[3]$, Case $\term = \term[2] \compo \term[3]$, Case $\term = \term[2]^{\kstar}$:
    By IH with the equations above.
    For example, when $\term = \term[2] \compo \term[3]$, we have
    \begin{align*}
        \hat{\val}(\term[2] \compo \term[3])  = \hat{\val}(\term[2]) \compo \hat{\val}(\term[3]) & = \hat{\val}(\lang{\term[2]}_{\vsig'}) \compo \hat{\val}(\lang{\term[3]}_{\vsig'}) \tag{IH}                                                    \\
                                                                                                 & = \hat{\val}(\lang{\term[2]}_{\vsig'} \compo \lang{\term[3]}_{\vsig'}) = \hat{\val}(\lang{\term[2] \compo \term[3]}_{\vsig'}). \tag*{\qedhere}
    \end{align*}
\end{proof}

\section{The identity inclusion problem}\label{section: identity}
We first consider the \emph{identity inclusion problem} w.r.t.\ languages:
\begin{center}
    Given a term $\term$, does $\LANG \models \eps \le \term$?
\end{center}
This problem is relatively easily solvable.
Since $\LANG \models \eps \le \term$ iff $\eps \in \hat{\val}(\term)$ for all \kl{language valuation} $\val$,
it suffices to consider the membership of the empty word $\eps$.
We use the following facts.
\begin{proposition}\label{proposition: identity}
    For all \kl{languages} $\la[1], \la[2]$, we have:
    \begin{align*}
        \eps \in \la[1] \union \la[2] & \quad\iff\quad \eps \in \la[1] \lor \eps \in \la[2]  \\
        \eps \in \la[1] \compo \la[2] & \quad\iff\quad \eps \in \la[1] \land \eps \in \la[2] \\
        \eps \in \la[1]^{\kstar}      & \quad\iff\quad \const{True}.
    \end{align*}
\end{proposition}
\begin{proof}
    Clear, by definition.
\end{proof}
\begin{lemma}\label{lemma: identity val}
    Let $\val, \val'$ be \kl{language valuations}.
    Assume that for all \kl{variables} $x$, $\eps \in \val(x)$ iff $\eps \in \val'(x)$.
    For all terms $\term$,
    \[\eps \in \hat{\val}(\term) \quad \iff \quad \eps \in \hat{\val}'(\term).\]
\end{lemma}
\begin{proof}
    By easy induction on $\term$ using \Cref{proposition: identity}.
    Case $\term = x, x^{\compl}$:
    Clear by the assumption.
    Case $\term$ is a constant:
    Trivial.
    For inductive cases, e.g., Case $\term = \term[2] \union \term[3]$:
    By using \Cref{proposition: identity}, we have
    \begin{align*}
        \eps \in \hat{\val}(\term[2] \union \term[3]) \iff \eps \in \hat{\val}(\term[2]) \lor \eps \in \hat{\val}(\term[3]) & \iff \eps \in \hat{\val}'(\term[2]) \lor \eps \in \hat{\val}'(\term[3]) \tag{IH} \\
                                                                                                                            & \iff \eps \in \hat{\val}'(\term[2] \union \term[3]).
    \end{align*}
    (Similarly for the other inductive cases.)
\end{proof}
By using \Cref{lemma: identity val}, it suffices to consider a finite number of valuations, as follows.
\begin{theorem}\label{theorem: identity val}
    For all terms $\term$, the following are equivalent:
    \begin{enumerate}
        \item \label{theorem: identity val 1} $\LANG \models \eps \le \term$ (i.e.,\ $\hat{\val}(\eps) \subseteq \hat{\val}(\term)$, for all \kl{language valuations} $\val$);
        \item \label{theorem: identity val 2} $\hat{\val}(\eps) \subseteq \hat{\val}(\term)$, for all \kl{language valuations} $\val$ over the empty set s.t.\  for all $x$, $\val(x) \subseteq \set{\const{I}}$.
    \end{enumerate}
\end{theorem}
\begin{proof}
    \ref{theorem: identity val 1}$\Rightarrow$\ref{theorem: identity val 2}:
    Trivial.
    \ref{theorem: identity val 2}$\Rightarrow$\ref{theorem: identity val 1}:
    We prove the contraposition.
    By $\LANG \not\models \eps \le \term$,
    there is a \kl{language valuation} $\val$ s.t.\ $\hat{\val}(\eps) \not\subseteq \hat{\val}(\term)$ (i.e., $\eps \not\in \hat{\val}(\term)$).
    Let $\val^{\tuple{}}$ be the \kl{language valuation} over the empty set defined by:
    \[\val^{\tuple{}}(x) \quad\defeq\quad \set{\eps \mid \eps \in \val(x)}.\]
    Then by \Cref{lemma: identity val}, $\eps \not\in \hat{\val}^{\tuple{}}(\term)$ holds; thus, we have $\hat{\val}^{\tuple{}}(\eps) \not\subseteq \hat{\val}^{\tuple{}}(\term)$.
\end{proof}

\begin{corollary}\label{corollary: identity decidable}
    The identity inclusion problem (given a term $\term$, does $\LANG \models \eps \le \term$?) is decidable and coNP-complete for KA terms with variable complements.
\end{corollary}
\begin{proof}
    (in coNP):
    \Cref{theorem: identity val} induces the following non-deterministic polynomial algorithm:
    \begin{enumerate}
        \item Pick up a \kl{language valuation} $\val$ over the empty set s.t.\  for all $x$, $\val(x) \subseteq \set{\const{I}}$, non-deterministically.
        \item If $\hat{\val}(\eps) \not\subseteq \hat{\val}(\term)$, then return $\const{True}$;
              otherwise return $\const{False}$.
    \end{enumerate}
    Then $\LANG \not\models \eps \le \term$ if some execution returns $\const{True}$;
    and $\LANG \models \eps \le \term$ otherwise.
    Hence, the identity inclusion problem is decidable in coNP (as its complemented problem is in NP).

    (coNP-hard):
    Because this problem subsumes the validity problem of propositional formulas in disjunctive normal form, which is a well-known coNP-complete problem \cite{cookComplexityTheoremprovingProcedures1971}.
    More precisely,
    given a propositional formula $\varphi$ in disjunctive normal form,
    let $\term$ be the term obtained from $\varphi$ by replacing each conjunction $\land$ with $\compo$ and each disjunction $\lor$ with $\union$ (where we map each positive literal $x$ to the variable $x$ and each negative literal $x^{-}$ to the complemented variable $x^{-}$); for example, if $\varphi = (x \land y^{-}) \lor (y \lor x^{-})$, then $\term = (x \compo y^{-}) \union (y \union x^{-})$.
    Then, for all \kl{language valuations} $\val$ (over the empty set s.t.\ for all $x, \val(x) \subseteq \set{\eps}$),
    we have: $\hat{\val}(\eps) \subseteq \hat{\val}(\term)$ holds iff $\varphi$ is $\const{True}$ on the valuation $\val'$,
    where $\val'$ is the map mapping each $x$ to $\const{True}$ if $\eps \in \val(x)$ and $\const{False}$ otherwise.
    Thus by \Cref{theorem: identity val}, $\LANG \models \eps \le \term$ iff $\varphi$ is valid.
    Hence, the identity inclusion problem is coNP-hard.
\end{proof}

\begin{remark}\label{remark: identity inclusion}
    Under the standard language valuation, the identity inclusion problem---given a term $\term$, does $\lang{\eps} \subseteq \lang{\term}$? (i.e., does $\eps \in \lang{\term}$?)---is decidable in P (because we can compute ``$\eps \in \lang{\term}$?'' by induction on $\term$, as $\eps \not\in \lang{x}$ and $\eps \in \lang{x^{\compl}}$ for every variable $x$).
    Hence, for KA terms with variable complements, the identity inclusion problem w.r.t.\ languages is strictly harder than that under the standard language valuation unless P = NP.
\end{remark}
\section{The variable inclusion problem and the universality problem}\label{section: univ}
Next, we consider the \emph{variable inclusion problem}:
\begin{center}
    Given a variable $x$ and a term $\term$, does $\LANG \models x \le \term$?
\end{center}
In the identity inclusion problem, if $\word \in \hat{\val}(\eps) \setminus \hat{\val}(\term)$, then $\word = \eps$ should hold; so it suffices to consider the membership of the empty word $\eps$.
However, in the variable inclusion problem, this situation changes: if $\word \in \hat{\val}(x) \setminus \hat{\val}(\term)$, then $\word$ is possibly any \kl{word}.

Nevertheless, we can overcome the problem above for KA terms with variable complements;
more precisely, from a \kl{language valuation} $\val$ s.t.\ $\word \in \hat{\val}(x) \setminus \hat{\val}(\term)$ for some \kl{word} $\word$,
we can construct a \kl{language valuation} $\val'$ s.t. $\ell \in \hat{\val}'(x) \setminus \hat{\val}'(\term)$ for some \emph{\kl{letter}} $\ell$.
If such $\val'$ can be constructed from $\val$,
then considering the membership of \kl{letters} suffices because we have
\begin{align*}
    \LANG \not\models x \le \term & \iff \word \in \hat{\val}(x) \setminus \hat{\val}(\term) \mbox{ for some \kl{language valuation} $\val$ and \kl{word} $\word$} \tag{By definition}                                                                                                  \\
                                  & \iff \ell \in \hat{\val}'(x) \setminus \hat{\val}'(\term) \mbox{ for some \kl{language valuation} $\val'$ and \emph{\kl{letter}} $\ell$} \tag{$\Longrightarrow$: By the condition of $\val'$. $\Longleftarrow$: Trivial by letting $\val = \val'$.}
\end{align*}
Such a \kl{language valuation} $\val'$ can be defined as follows:
\begin{definition}\label{definition: val for word to letter}
    For a \kl{language valuation} $\val$ over a set $X$ and a \kl{word} $\word$ over $X$,
    the \kl{language valuation} $\val^{\word}$ over the set $\set{\ell}$ (where $\ell$ is a \kl{letter}) is defined as follows:
    \[\val^{\word}(x) \quad\defeq\quad \set{\eps \mid \eps \in \val(x)} \union \set{\ell \mid \word \in \val(x)}.\]
\end{definition}
In the following, we prove that $\val^{\word}$ satisfies the condition of $\val'$ above, that is, the following conditions:
\begin{itemize}
    \item $\word \in \hat{\val}(x) \Longrightarrow \ell \in \hat{\val}^{\word}(x)$;
    \item $\word \not\in \hat{\val}(\term) \Longrightarrow \ell \not\in \hat{\val}^{\word}(\term)$.
\end{itemize}
The first condition is clear by the definition of $\val^{\word}$.
We prove the second condition in \Cref{lemma: ell abstraction}.
We prepare the following fact:
\begin{proposition}[cf.\ \Cref{proposition: identity}]\label{proposition: letter}
    For all \kl{languages} $\la[1], \la[2]$ and \kl{letters} $a$, we have:
    \begin{align*}
        a \in \la[1] \union \la[2] & \quad\iff\quad a \in \la[1] \lor a \in \la[2]                                                          \\
        a \in \la[1] \compo \la[2] & \quad\iff\quad (a \in \la[1] \land \const{I} \in \la[2]) \lor (\const{I} \in\la[1] \land a \in \la[2]) \\
        a \in \la[1]^{\kstar}      & \quad\iff\quad a \in \la[1].
    \end{align*}
\end{proposition}
\begin{proof}
    Clear, by definition.
\end{proof}
\begin{lemma}[cf.\ \Cref{lemma: identity val}]\label{lemma: ell abstraction}
    Let $\val$ be a \kl{language valuation} and $\word$ be a \kl{word}.
    For all terms $\term$, we have:
    \[\ell \in \hat{\val}^{\word}(\term) \quad \Longrightarrow \quad \word \in \hat{\val}(\term).\]
\end{lemma}
\begin{proof}
    By induction on $\term$.

    Case $\term = x, x^{-}$:
    By the construction of $\val^{\word}$,
    $\ell \in \hat{\val}^{\word}(x)$ iff $\word \in \hat{\val}(x)$.
    (Hence, we also have $\ell \in \hat{\val}^{\word}(x^{-})$ iff $\word \in \hat{\val}(x^{-})$.)

    Case $\term = \bot, \eps$:
    By $\ell \not\in \hat{\val}^{\word}(\term)$.

    Case $\term = \term[2] \union \term[3]$:
    We have:
    \begin{align*}
        \ell \in \hat{\val}^{\word}(\term[2]) \union \hat{\val}^{\word}(\term[3]) & \Longleftrightarrow \ell \in \hat{\val}^{\word}(\term[2]) \lor \ell \in \hat{\val}^{\word}(\term[3]) \tag{\Cref{proposition: letter}} \\
                                                                                  & \Longrightarrow \word \in \hat{\val}(\term[2]) \lor \word \in \hat{\val}(\term[3]) \tag{IH}                                           \\
                                                                                  & \Longrightarrow \word \in \hat{\val}(\term[2]) \union \hat{\val}(\term[3]).
    \end{align*}

    Case $\term = \term[2] \compo \term[3]$:
    We have:
    \begin{align*}
        \ell \in \hat{\val}^{\word}(\term[2]) \compo \hat{\val}^{\word}(\term[3]) & \Longleftrightarrow (\ell \in \hat{\val}^{\word}(\term[2]) \land \eps \in \hat{\val}^{\word}(\term[3])) \lor (\eps \in \hat{\val}^{\word}(\term[2]) \land \ell \in \hat{\val}^{\word}(\term[3]))          \tag{\Cref{proposition: letter}} \\
                                                                                  & \Longrightarrow (\word \in \hat{\val}(\term[2]) \land \eps \in \hat{\val}(\term[3])) \lor (\eps \in \hat{\val}(\term[2]) \land \word \in \hat{\val}(\term[3])) \tag{IH with \Cref{lemma: identity val}}                                    \\
                                                                                  & \Longrightarrow \word \in \hat{\val}(\term[2]) \compo \hat{\val}(\term[3])
    \end{align*}

    Case $\term = \term[2]^{\kstar}$:
    We have:
    \begin{align*}
        \ell \in \hat{\val}^{\word}(\term[2])^{\kstar} & \Longleftrightarrow \ell \in \hat{\val}^{\word}(\term[2]) \tag{\Cref{proposition: letter}} \\
                                                       & \Longrightarrow \word \in \hat{\val}(\term[2]) \tag{IH}                                    \\
                                                       & \Longrightarrow \word \in \hat{\val}(\term[2])^{\kstar}
    \end{align*}
\end{proof}
Thus we have obtained the expected condition for $\val^{\word}$ as follows:
\begin{corollary}\label{corollary: ell abstraction variable}
    Let $\val$ be a \kl{language valuation} and $\word$ be a \kl{word}.
    For all \kl{variables} $x$ and terms $\term[1]$,
    \[\word \in \hat{\val}(x) \setminus \hat{\val}(\term[1]) \quad \Longrightarrow \quad  \ell \in \hat{\val}^{\word}(x) \setminus \hat{\val}^{\word}(\term[1]).\]
\end{corollary}
\begin{proof}
    For $\ell \in \hat{\val}^{\word}(x)$:
    By the construction of $\val^{\word}$,
    $\ell \in \hat{\val}^{\word}(x)$ iff $\word \in \hat{\val}(x)$.
    For $\ell \not\in \hat{\val}^{\word}(\term[1])$:
    By \Cref{lemma: ell abstraction}.
\end{proof}
\begin{theorem}\label{theorem: char and univ variable}
    For all \kl{variables} $x$ and terms $\term[1]$, the following are equivalent:
    \begin{enumerate}
        \item \label{theorem: char and univ variable 1} $\LANG \models x \le \term[1]$;
        \item \label{theorem: char and univ variable 2} $\hat{\val}(x) \subseteq \hat{\val}(\term[1])$ for all \kl{language valuations} $\val$ over the set $\set{\ell}$ s.t.\ $\val(y) \subseteq \set{\const{I}, \ell}$ for all $y$;
        \item \label{theorem: char and univ variable 3} $\hat{\val}^{\word}(x) \subseteq \hat{\val}^{\word}(\term[1])$ for all \kl{language valuations} $\val$ and \kl{words} $\word$.
    \end{enumerate}
\end{theorem}
\begin{proof}
    \ref{theorem: char and univ variable 1}$\Rightarrow$\ref{theorem: char and univ variable 2},
    \ref{theorem: char and univ variable 2}$\Rightarrow$\ref{theorem: char and univ variable 3}:
    Trivial, as $\hat{\val}^{\word}(y) \subseteq \set{\eps, \ell}$ for all $y$.
    \ref{theorem: char and univ variable 3}$\Rightarrow$\ref{theorem: char and univ variable 1}:
    The contraposition is shown by \Cref{corollary: ell abstraction variable}.
\end{proof}
\begin{corollary}\label{corollary: variable decidable}
    The variable inclusion problem (given a variable $x$ and a term $\term$, does $\LANG \models x \le \term$?) is decidable and coNP-complete for KA terms with variable complements.
\end{corollary}
\begin{proof}
    (in coNP):
    By the condition \ref{theorem: char and univ variable 2} of \Cref{theorem: char and univ variable},
    we can give an algorithm as with \Cref{corollary: identity decidable}.
    (coNP-hard):
    We give a reduction from the validity problem of propositional formulas in disjunctive normal form, as with \Cref{corollary: identity decidable}.
    Given a propositional formula $\varphi$ in disjunctive normal form,
    let $\term$ be the term obtained by the translation in \Cref{corollary: identity decidable}; so, we have that $\varphi$ is valid iff $\LANG \models \eps \le \term$.
    Then we also have that $\LANG \models \eps \le \term$ iff $\LANG \models z \le z \compo \term$ (where $z$ is a fresh \kl{variable});
    the direction $\Longrightarrow$ is shown by the congruence law, and the direction $\Longleftarrow$ is shown by the substitution law.
    Therefore, we have that $\varphi$ is valid iff $\LANG \models z \le z \compo \term$; thus, the variable inclusion problem is coNP-hard.
\end{proof}

\subsection{Generalization from variables to composition-free terms}\label{section: composition-free}
The proof above applies to not only \kl{variables} but also terms $\term$ having the following property:
For all \kl{language valuations} $\val$,
\begin{flalign*}
    \label{cond: length 1} \mbox{for all non-empty \kl{words} $\word$, \qquad} \word \in \hat{\val}(\term) \Longrightarrow \ell \in \hat{\val}^{\word}(\term). \tag*{$(\LENGTH{1})$}
\end{flalign*}
(This condition is intended for \kl{composition-free} terms (\Cref{lemma: composition-free}). This is generalized to $(\LENGTH{n})$ in \Cref{section: star-free}.)
If $\term[1]$ satisfies the condition \ref{cond: length 1}, then combining with \Cref{lemma: ell abstraction}
(and with \Cref{lemma: identity val} for the empty word $\eps$) yields that
for all \kl{language valuations} $\val$ and \kl{words} $\word$,
\begin{align*}
    \word \in \hat{\val}(\term[1]) \setminus \hat{\val}(\term[2]) & \quad \Longrightarrow \quad \begin{cases}
                                                                                                    \ell \in \hat{\val}^{\word}(\term[1]) \setminus \hat{\val}^{\word}(\term[2]) & (\mbox{if $\word \neq \eps$}) \\
                                                                                                    \eps \in \hat{\val}^{\word}(\term[1]) \setminus \hat{\val}^{\word}(\term[2]) & (\mbox{if $\word = \eps$})
                                                                                                \end{cases}.
\end{align*}
Hence, we have the following:
\begin{theorem}[cf.\ \Cref{theorem: char and univ variable}]\label{theorem: char and univ}
    For all terms $\term[1], \term[2]$, if $\term[1]$ satisfies \ref{cond: length 1}, then the following are equivalent:
    \begin{enumerate}
        \item \label{theorem: char and univ 1} $\LANG \models \term[1] \le \term[2]$;
        \item \label{theorem: char and univ 2} $\hat{\val}(\term[1]) \subseteq \hat{\val}(\term[2])$ for all \kl{language valuations} $\val$ over the set $\set{\ell}$ s.t.\ $\val(x) \subseteq \set{\const{I}, \ell}$ for all $x$;
        \item \label{theorem: char and univ 3} $\hat{\val}^{\word}(\term[1]) \subseteq \hat{\val}^{\word}(\term[2])$ for all \kl{language valuations} $\val$ and \kl{words} $\word$.
    \end{enumerate}
\end{theorem}
\begin{proof}
    As with \Cref{theorem: char and univ variable} (use the above, instead of \Cref{corollary: ell abstraction variable}).
\end{proof}
\Cref{theorem: char and univ} can apply to \kl{composition-free} terms.
We say that a term $\term$ is \intro*\kl{composition-free} if composition ($\compo$) nor Kleene-star ($\bl^{\kstar}$) does not occur in $\term$.
\begin{lemma}\label{lemma: composition-free}
    Every \kl{composition-free} terms $\term$ satisfies the condition \ref{cond: length 1}.
\end{lemma}
\begin{proof}
    By easy induction on $\term$.
    Case $\term = x, x^{\compl}$:
    By the definition of $\val^{\word}$.
    Case $\term = \eps$:
    By that $\word \not\in \hat{\val}(\eps)$ holds for all non-empty \kl{words} $\word$.
    Case $\term = \bot$:
    By $\word \not\in \hat{\val}(\bot)$ always.
    Case $\term = \term[2] \union \term[3]$:
    By IH, we have that $\word \in \hat{\val}(\term[2]) \union \hat{\val}(\term[3]) \iff \word \in \hat{\val}(\term[2]) \lor \word \in \hat{\val}(\term[3]) \Longrightarrow \ell \in \hat{\val}^{\word}(\term[2]) \lor \ell \in \hat{\val}^{\word}(\term[3]) \iff \ell \in \hat{\val}^{\word}(\term[2]) \union \hat{\val}^{\word}(\term[3])$.
\end{proof}
\begin{corollary}\label{corollary: composition-free decidable}
    The following problem is coNP-complete for KA terms with variable complements:
    \begin{center}
        Given a \kl{composition-free} term $\term[1]$ and a term $\term[2]$, does $\LANG \models \term[1] \le \term[2]$ hold?
    \end{center}
\end{corollary}
\begin{proof}
    (coNP-hard):
    By \Cref{corollary: identity decidable}, as $\term[1]$ is possibly $\eps$.
    (in coNP):
    By \Cref{theorem: char and univ} with \Cref{lemma: composition-free}, we can give an algorithm (from the condition \ref{theorem: char and univ 2} of \Cref{theorem: char and univ}) as with \Cref{corollary: variable decidable}.
\end{proof}

\begin{remark}\label{remark: composition-free}
    \Cref{lemma: composition-free} fails for non-\kl{composition-free} terms.
    For example, when $\val(x) = \set{a}$, we have
    \begin{align*}
        a a \in \hat{\val}(x x) &  & \ell \not\in \hat{\val}^{a a}(x x).
    \end{align*}
    (Note that $\hat{\val}^{a a}(x x) = \emptyset$, as $\hat{\val}^{a a}(x) = \emptyset$ by $\val(x) = \set{a}$.)
\end{remark}

\subsection{The universality problem}
The \emph{universality problem} is the following problem:
\begin{center}
    Given a term $\term$, does $\LANG \models \top \le \term$?
\end{center}
As a consequence of \Cref{corollary: composition-free decidable}, the universality problem is also decidable and coNP-complete.
\begin{corollary}\label{corollary: universality decidable}
    The universality problem is decidable and coNP-complete for KA terms with variable complements.
\end{corollary}
\begin{proof}
    (in coNP):
    We can apply \Cref{corollary: composition-free decidable} because the term $x \union x^{\compl}$ is \kl{composition-free} and $\LANG \models \top = x \union x^{\compl}$ holds.
    (coNP-hard):
    Similar to \Cref{corollary: variable decidable}.
    Given a propositional formula $\varphi$ in disjunctive normal form,
    let $\term$ be the term obtained by the translation in \Cref{corollary: identity decidable}; so, we have that $\varphi$ is valid iff $\LANG \models \eps \le \term$.
    Then we also have that $\LANG \models \eps \le \term$ iff $\LANG \models \top \le \top \compo \term$, which is proved as follows.
    $\Longrightarrow$:
    By the congruence law.
    $\Longleftarrow$:
    We prove the contraposition.
    Assume $\LANG \not\models \eps \le \term$; then $\eps \not\in \hat{\val}(\term)$ for some \kl{language valuation} $\val$.
    Then $\eps \not\in \hat{\val}(\top \compo \term)$ holds; thus, $\LANG \not\models \top \le \top \compo \term$.
    Hence, the universality problem is coNP-complete.
\end{proof}

\begin{remark}
    In the standard language equivalence, because $\lang{\vsig^{\kstar}} = \lang{\top}$ (and the constant $\top$ is usually not a primitive symbol of regular expressions),
    the universality problem is always of the form: $\lang{\vsig^{\kstar}} = \lang{\term}$.
    However, $\LANG \models \vsig^{\kstar} \le \term$ is different from $\LANG \models \top \le \term$, as $\LANG \not \models \vsig^{\kstar} = \top$.
\end{remark}

\begin{remark}\label{remark: universality}
    Under the standard language equivalence, the universality problem---given a term $\term$, does $\lang{\top} \subseteq \lang{\term}$? (i.e., does $\lang{\term} = \vsig^*$?)---is PSPACE-hard \cite{meyerEquivalenceProblemRegular1972, Meyer1973, Hunt1976}.
    Hence, for KA terms with variable complements, the universality problem w.r.t.\ languages is strictly easier (cf.\ \Cref{remark: identity inclusion}) than that under the standard language equivalence unless NP = PSPACE.
\end{remark}

\section{The word inclusion problem}\label{section: word}
Let $\vsig' = \set{x, x^{\compl} \mid x \in \vsig}$.
The \emph{word inclusion problem} is the following problem:
\begin{center}
    Given a \kl{word} $\word$ over $\vsig'$ and a term $\term$, does $\LANG \models \word \le \term$?
\end{center}
As \Cref{remark: composition-free} shows, we cannot apply the method given in \Cref{section: univ} straightforwardly.
Nevertheless, we can solve this problem by generalizing the \kl{language valuation} of \Cref{definition: val for word to letter}, as follows.
The valuations in \Cref{definition: val for word to letter,definition: val for words to letters} are given by the first author. %
\begin{definition}[\intro*\kl{words-to-letters valuations}]\label{definition: val for words to letters}
    For a \kl{language valuation} $\val$ over a set $X$ and \kl{words} $\word_0, \dots, \word_{n-1}$ over $X$ (where $n \ge 0$),
    the \kl{language valuation} $\val^{\tuple{\word_0, \dots, \word_{n-1}}}$ over the set $\set{\ell_0, \dots, \ell_{n-1}}$ (where $\ell_{0}, \dots, \ell_{n-1}$ are pairwise distinct \kl{letters}) is defined as follows:
    \[\val^{\tuple{\word_0, \dots, \word_{n-1}}}(x) \quad\defeq\quad \set{\ell_i \dots \ell_{j-1} \mid 0 \le i \le j \le n \  \land \  \word_i \dots \word_{j-1} \in \val(x)}.\]
\end{definition}
\noindent (Note that the \kl{language valuation} $\val^{\word}$ (\Cref{definition: val for word to letter}) is the case $n = 1$ of \Cref{definition: val for words to letters} and the \kl{language valuation} $\val^{\tuple{}}$ in the proof of \Cref{theorem: identity val} is the case $n = 0$ of \Cref{definition: val for words to letters}.)

By using \kl{words-to-letters valuations}, we can naturally strengthen the results in \Cref{section: univ} from \kl{variables} to \kl{words}.
We prepare the following fact:
\begin{proposition}[cf.\ \Cref{proposition: letter}]\label{proposition: word}
    For all \kl{languages} $\la[1], \la[2]$ and \kl{words} $\word$,
    \begin{align*}
        \word \in \la[1] \union \la[2] & \quad\iff\quad \word \in \la[1] \lor \word \in \la[2]                                                                                                                            \\
        \word \in \la[1] \compo \la[2] & \quad\iff\quad \exists \word[2], \word[2]' \mbox{ s.t.\ } \word = \word[2] \word[2]',\  \word[2] \in \la[1] \land \word[2]' \in \la[2]                                           \\
        \word \in \la[1]^{\kstar}      & \quad\iff\quad \exists n \in \nat, \exists \word[2]_{0}, \dots, \word[2]_{n-1} \mbox{ s.t.\ } \word = \word[2]_{0} \dots \word[2]_{n-1}, \forall i < n,\  \word[2]_i \in \la[1].
    \end{align*}
\end{proposition}
\begin{proof}
    By definition.
\end{proof}

\begin{lemma}[cf.\ \Cref{lemma: ell abstraction}]\label{lemma: ell abstraction gen}
    Let $\val$ be a \kl{language valuation} and $\word_0, \dots, \word_{n-1}$ be \kl{words} (where $n \ge 0$).
    For all terms $\term$ and $0 \le i \le j \le n$, we have:
    \[\ell_{i} \dots \ell_{j-1} \in \hat{\val}^{\tuple{\word_0, \dots, \word_{n-1}}}(\term) \quad \Longrightarrow \quad \word_{i} \dots \word_{j-1} \in \hat{\val}(\term).\]
\end{lemma}
\begin{proof}
    By induction on $\term$.

    Case $\term = x, x^{\compl}$:
    By the construction of $\val^{\tuple{\word_0, \dots, \word_{n-1}}}$,
    $\ell_{i} \dots \ell_{j-1} \in \hat{\val}^{\tuple{\word_0, \dots, \word_{n-1}}}(x)$ iff $\word_i \dots \word_{j-1} \in \hat{\val}(x)$.
    (Hence, we also have $\ell_{i} \dots \ell_{j-1} \in \hat{\val}^{\tuple{\word_0, \dots, \word_{n-1}}}(x^{\compl})$ iff $\word_{i}\dots\word_{j-1} \in \hat{\val}(x^{\compl})$.)

    Case $\term = \bot$, Case $\term = \eps$ where $i < j$:
    By $\ell_{i} \dots \ell_{j-1} \not\in \hat{\val}^{\tuple{\word_0, \dots, \word_{n-1}}}(\term)$.

    Case $\term = \eps$ where $i = j$:
    By $\eps \in \hat{\val}(\eps)$.

    Case $\term = \term[2] \union \term[3]$:
    We have
    \begin{align*}
        \ell_{i} \dots \ell_{j-1} \in \hat{\val}^{\tuple{\word_0, \dots, \word_{n-1}}}(\term[2] \union \term[3]) & \Longleftrightarrow \ell_{i} \dots \ell_{j-1} \in \hat{\val}^{\tuple{\word_0, \dots, \word_{n-1}}}(\term[2]) \lor \ell_{i} \dots \ell_{j-1} \in \hat{\val}^{\tuple{\word_0, \dots, \word_{n-1}}}(\term[3]) \tag{\Cref{proposition: word}} \\
                                                                                                                 & \Longrightarrow \word_{i} \dots \word_{j-1} \in \hat{\val}(\term[2]) \lor \word_{i} \dots \word_{j-1} \in \hat{\val}(\term[3]) \tag{IH}                                                                                                   \\
                                                                                                                 & \Longrightarrow \word_{i} \dots \word_{j-1} \in \hat{\val}(\term[2] \union \term[3]).
    \end{align*}

    Case $\term = \term[2] \compo \term[3]$:
    We have
    \begin{align*}
        \ell_{i} \dots \ell_{j-1} \in \hat{\val}^{\tuple{\word_0, \dots, \word_{n-1}}}(\term[2] \compo \term[3]) & \Longleftrightarrow \bigvee_{i \le k \le j} (\ell_{i} \dots \ell_{k-1} \in \hat{\val}^{\tuple{\word_0, \dots, \word_{n-1}}}(\term[2]) \land \ell_{k} \dots \ell_{j-1} \in \hat{\val}^{\tuple{\word_0, \dots, \word_{n-1}}}(\term[3])) \tag{\Cref{proposition: word}} \\
                                                                                                                 & \Longrightarrow \bigvee_{i \le k \le j} (\word_{i} \dots \word_{k-1} \in \hat{\val}(\term[2]) \land \word_{k} \dots \word_{j-1} \in \hat{\val}(\term[3])) \tag{IH}                                                                                                   \\
                                                                                                                 & \Longrightarrow \word_{i} \dots \word_{j-1} \in \hat{\val}(\term[2] \compo \term[3]).
    \end{align*}

    Case $\term = \term[2]^*$:
    We have
    \begin{align*}
        \ell_{i} \dots \ell_{j-1} \in \hat{\val}^{\tuple{\word_0, \dots, \word_{n-1}}}(\term[2]^*) & \Longleftrightarrow \exists m \in \nat, \bigvee_{i = k_0 \le k_1 \le \dots \le k_m = j} \bigwedge_{l = 1}^{m} (\ell_{k_{l-1}} \dots \ell_{k_{l}-1} \in \hat{\val}^{\tuple{\word_0, \dots, \word_{n-1}}}(\term[2])) \tag{\Cref{proposition: word}} \\
                                                                                                   & \Longrightarrow \exists m \in \nat, \bigvee_{i = k_0 \le k_1 \le \dots \le k_m = j} \bigwedge_{l = 1}^{m} (\word_{k_{l-1}} \dots \word_{k_{l}-1} \in \hat{\val}(\term[2])) \tag{IH}                                                               \\
                                                                                                   & \Longrightarrow \word_{i} \dots \word_{j-1} \in \hat{\val}(\term[2]^*).
    \end{align*}
\end{proof}
\begin{corollary}[cf.\ \Cref{corollary: ell abstraction variable}]\label{corollary: ell abstraction word}
    Let $\val$ be a \kl{language valuation}, $\word$ be a \kl{word},
    $\word_{0}, \dots, \word_{n-1}$ be \kl{words} s.t.\ $\word = \word_{0} \dots \word_{n-1}$.
    For all \kl{words} $\word[2]$ over $\vsig'$ of \kl{length} $n$ and all terms $\term$,
    \[\word \in \hat{\val}(\word[2]) \setminus \hat{\val}(\term) \quad \Longrightarrow \quad  \ell_0 \dots \ell_{n-1} \in \hat{\val}^{\tuple{\word_0, \dots, \word_{n-1}}}(\word[2]) \setminus \hat{\val}^{\tuple{\word_0, \dots, \word_{n-1}}}(\term).\]
\end{corollary}
\begin{proof}
    For $\ell_0 \dots \ell_{n-1} \in \hat{\val}^{\tuple{\word_0, \dots, \word_{n-1}}}(\word[2])$:
    Let $\word[2] = x_0 \dots x_{n-1}$.
    For each $i < n$,
    by the construction of $\val^{\tuple{\word_0, \dots, \word_{n-1}}}$,
    $\ell_i \in \hat{\val}^{\tuple{\word_0, \dots, \word_{n-1}}}(x_i)$ iff $\word_i \in \hat{\val}(x_i)$.
    Thus, we have that $\ell_0 \dots \ell_{n-1} \in \hat{\val}^{\tuple{\word_0, \dots, \word_{n-1}}}(\word[2])$.
    For $\ell_0 \dots \ell_{n-1} \not\in \hat{\val}^{\tuple{\word_0, \dots, \word_{n-1}}}(\term)$:
    By \Cref{lemma: ell abstraction gen}.
\end{proof}
\begin{theorem}[cf.\ \Cref{theorem: char and univ variable}]\label{theorem: char and univ word}
    For all \kl{words} $\word[2]$ over $\vsig'$ of \kl{length} $n$ and all terms $\term$,
    the following are equivalent:
    \begin{enumerate}
        \item \label{theorem: char and univ word 1} $\LANG \models \word[2] \le \term$;
        \item \label{theorem: char and univ word 2} $\hat{\val}(\word[2]) \subseteq \hat{\val}(\term)$ for all \kl{language valuations} $\val$ s.t.\ $\val(x) \subseteq \set{\ell_{i} \dots \ell_{j} \mid 0 \le i \le j \le n}$ for all $x$;
        \item \label{theorem: char and univ word 3} $\hat{\val}^{\tuple{\word_0, \dots, \word_{n-1}}}(\word[2]) \subseteq \hat{\val}^{\tuple{\word_0, \dots, \word_{n-1}}}(\term)$ for all \kl{language valuations} $\val$ and \kl{words} $\word_0, \dots, \word_{n-1}$.
    \end{enumerate}
\end{theorem}
\begin{proof}
    \ref{theorem: char and univ word 1}$\Rightarrow$\ref{theorem: char and univ word 2}, \ref{theorem: char and univ word 2}$\Rightarrow$\ref{theorem: char and univ word 3}:
    Trivial.
    \ref{theorem: char and univ word 3}$\Rightarrow$\ref{theorem: char and univ word 1}:
    The contraposition is shown by \Cref{corollary: ell abstraction word}.
\end{proof}

\begin{corollary}[cf.\ \Cref{corollary: variable decidable}]\label{corollary: word decidable}
    The word inclusion problem (given a \kl{word} $\word$ and a term $\term$, does $\LANG \models \word \le \term$?) is decidable and coNP-complete for KA terms with variable complements.
\end{corollary}
\begin{proof}
    (coNP-hard):
    By \Cref{corollary: identity decidable}, as $\word$ is possibly $\const{I}$.
    (in coNP):
    By the condition \ref{theorem: char and univ word 2} of \Cref{theorem: char and univ word},
    we can give an algorithm as with \Cref{corollary: variable decidable}.
\end{proof}

\subsection{Generalization from words to star-free terms}\label{section: star-free}
We can apply \Cref{theorem: char and univ word} to not only \kl{words} over $\vsig'$ but also terms $\term$ having the following property:
\begin{flalign*}
                           & \mbox{For all \kl{language valuations} $\val$ and non-empty \kl{words} $\word$, for some $\word_{0}, \dots, \word_{n-1}$ s.t.\ $\word = \word_{0} \dots \word_{n-1}$,} \\
    \label{cond: length n} & \word \in \hat{\val}(\term) \Longrightarrow \ell_{0} \dots \ell_{n-1} \in \hat{\val}^{\tuple{\word_{0}, \dots, \word_{n-1}}}(\term). \tag*{$(\LENGTH{n})$}
\end{flalign*}
If $\term[1]$ satisfies the condition \ref{cond: length n}, then combining with \Cref{lemma: ell abstraction gen}
(and with \Cref{lemma: identity val} for the empty word $\eps$) yields that
for all \kl{language valuations} $\val$ and \kl{words} $\word$, for some \kl{words} $\word_{0}, \dots, \word_{n-1}$ s.t.\ $\word = \word_{0} \dots \word_{n-1}$, we have:
\begin{align*}
    \word \in \hat{\val}(\term[1]) \setminus \hat{\val}(\term[2]) & \quad \Longrightarrow \quad \begin{cases}
                                                                                                    \ell_{0} \dots \ell_{n-1} \in \hat{\val}^{\tuple{\word_{0}, \dots, \word_{n-1}}}(\term[1]) \setminus \hat{\val}^{\tuple{\word_{0}, \dots, \word_{n-1}}}(\term[2]) & (\mbox{if $\word \neq \eps$}) \\
                                                                                                    \eps \in \hat{\val}^{\tuple{\word_{0}, \dots, \word_{n-1}}}(\term[1]) \setminus \hat{\val}^{\tuple{\word_{0}, \dots, \word_{n-1}}}(\term[2])                      & (\mbox{if $\word = \eps$})
                                                                                                \end{cases}.
\end{align*}

Hence, we have the following:
\begin{theorem}[cf.\ \Cref{theorem: char and univ}]\label{theorem: char and univ gen}
    For all terms $\term[1], \term[2]$,
    if $\term[1]$ satisfies \ref{cond: length n}, the following are equivalent:
    \begin{enumerate}
        \item \label{theorem: char and univ gen 1} $\LANG \models \term[1] \le \term[2]$;
        \item \label{theorem: char and univ gen 2} $\hat{\val}(\term[1]) \subseteq \hat{\val}(\term[2])$ for all \kl{language valuations} $\val$ over the set $\set{\ell_{0}, \dots, \ell_{n-1}}$ s.t.\ $\val(x) \subseteq  \set{\ell_{i} \dots \ell_{j} \mid 0 \le i \le j \le n}$ for all $x$;
        \item \label{theorem: char and univ gen 3} $\hat{\val}^{\tuple{\word_0, \dots, \word_{n-1}}}(\term[1]) \subseteq \hat{\val}^{\tuple{\word_0, \dots, \word_{n-1}}}(\term[2])$ for all \kl{language valuations} $\val$ and \kl{words} $\word_0, \dots, \word_{n-1}$.
    \end{enumerate}
\end{theorem}
\begin{proof}
    As with \Cref{theorem: char and univ word} (use the above, instead of \Cref{corollary: ell abstraction word}).
\end{proof}
By using \Cref{theorem: char and univ gen}, we can generalize \Cref{corollary: word decidable} from \kl{words} to \kl{star-free} terms.
We say that a term $\term$ is \intro*\kl{star-free} if the Kleene-star ($\bl^{\kstar}$) does not occur in $\term$.
\begin{lemma}[cf.\ \Cref{lemma: composition-free}]\label{lemma: star-free}
    Every \kl{star-free} term $\term$ satisfies \ref{cond: length n} for some $n$.
\end{lemma}
\begin{proof}
    Because the set $\lang{\term}_{\vsig'}$ is finite as $\term$ is \kl{star-free},
    let $n$ be the maximal \kl{length} among \kl{words} in $\lang{\term}_{\vsig'}$.
    Let $\val$ be a \kl{language valuation} and let $\word$ be a non-empty \kl{word} such that $\word \in \hat{\val}(\term)$.
    Since $\hat{\val}(\term) = \hat{\val}(\lang{\term}_{\vsig'})$ (\Cref{lemma: lang val}),
    there is a \kl{word} $\word[2] \in \lang{\term}_{\vsig'}$ such that $\word \in \hat{\val}(\word[2])$.
    Let $\word[2] = x_0 \dots x_{m-1}$ (note that $m \ge 1$, as $\word$ is non-empty and $\word \in \hat{\val}(\word[2])$).
    Since $\word \in \hat{\val}(x_0 \dots x_{m-1})$,
    there are $\word_{0}, \dots, \word_{m-1}$ of $\word = \word_{0} \dots \word_{m-1}$ such that $\word_{i} \in \hat{\val}(x_{i})$ for every $i$.
    Let $\val' \defeq \val^{\tuple{\word_{0}, \dots, \word_{m-1}, \eps, \dots, \eps}}$, where the length of the sequence is $n$.
    Then, we have $\ell_{i} \in \hat{\val}'(x_i)$ for every $0 \le i \le m-2$
    and $\ell_{m-1}\ell_{m} \dots \ell_{n-1} \in \hat{\val}'(x_{m-1})$;
    thus, $\ell_{0} \dots \ell_{n-1} \in \hat{\val}'(x_{0} \dots x_{m-1}) = \hat{\val}'(\word[2]) \subseteq \hat{\val}'(\term)$.
    Hence, this completes the proof.
\end{proof}

\begin{corollary}\label{corollary: star-free}
    The following problem is coNP-complete for KA terms with variable complements:
    \begin{center}
        Given a \kl{star-free} term $\term[1]$ and a term $\term[2]$, does $\LANG \models \term[1] \le \term[2]$?
    \end{center}
\end{corollary}
\begin{proof}
    (coNP-hard):
    By \Cref{corollary: identity decidable}, as $\term[1]$ is possibly $\eps$.
    (in coNP):
    By \Cref{lemma: star-free}, we can give an algorithm as with \Cref{corollary: word decidable}.
\end{proof}

\subsection{words-to-letters valuation property}
Finally, we show the following property; thus, we have that \kl{words-to-letters valuations} are sufficient for the equational theory of (full) KA terms with variable complements.
\begin{corollary}[words-to-letters valuation property]\label{corollary: word witness}
    For all terms $\term[1], \term[2]$, the following are equivalent:
    \begin{enumerate}
        \item \label{corollary: word witness 1} $\LANG \not\models \term[1] \le \term[2]$;
        \item \label{corollary: word witness 3} there is a \kl{words-to-letters valuation} $\val$ such that $\hat{\val}(\term[1]) \not\subseteq \hat{\val}(\term[2])$.
    \end{enumerate}
\end{corollary}
\begin{proof}
    \ref{corollary: word witness 3}$\Rightarrow$\ref{corollary: word witness 1}:
    Trivial.
    \ref{corollary: word witness 1}$\Rightarrow$\ref{corollary: word witness 3}:
    Since $\LANG \not\models \term[1] \le \term[2]$, there is a \kl{language valuation} $\val$ such that $\hat{\val}(\term[1]) \not\subseteq \hat{\val}(\term[2])$.
    Since $\hat{\val}(\term[1]) = \hat{\val}(\lang{\term[1]}_{\vsig'})$ (\Cref{lemma: lang val}),
    there is a \kl{word} $\word[2] \in \lang{\term[1]}_{\vsig'}$ such that $\hat{\val}(\word[2]) \not\subseteq \hat{\val}(\term[2])$ (i.e., $\LANG \not\models \word[2] \le \term[2]$).
    Let $n$ be the length of $\word[2]$.
    By \Cref{theorem: char and univ word},
    there are a \kl{words-to-letters valuation} $\val$ and \kl{words} $\word_0, \dots, \word_{n-1}$ such that
    $\hat{\val}^{\tuple{\word_0, \dots, \word_{n-1}}}(\word[2]) \not\subseteq \hat{\val}^{\tuple{\word_0, \dots, \word_{n-1}}}(\term[2])$.
    Thus
    $\hat{\val}^{\tuple{\word_0, \dots, \word_{n-1}}}(\term[1]) \not\subseteq \hat{\val}^{\tuple{\word_0, \dots, \word_{n-1}}}(\term[2])$, as $\word[2] \in \lang{\term[1]}_{\vsig'}$ (\Cref{lemma: lang val}).
    Hence this completes the proof.
\end{proof}

\section{On the equational theory of words with variable complements}\label{section: completeness word}
We prove that the equational theory of \kl{words} over $\vsig'$ coincides with the \kl{word} equivalence (\Cref{theorem: completeness word}).
We give \kl{language valuations} for separating two distinct \kl{words} based on \kl{words-to-letters valuations}.
\begin{lemma}\label{lemma: completeness word}
    Let $\word[1] = x_0 \dots x_{n-1}$ and $\word[2] = y_0 \dots y_{m-1}$ be \kl{words} over $\vsig'$, where $n \le m$.
    Let $\val$ be a \kl{language valuation} over $\set{\ell_{0}, \dots, \ell_{n-1}}$ such that
    \begin{itemize}
        \item for all $i < n$, $\ell_{i} \in \hat{\val}(x_i)$;
        \item for all $i < m$ and $i \le j \le n$,
              $\ell_{i} \dots \ell_{j-1} \in \hat{\val}(y_i)$ iff $(y_i = x_i \land j = i+1)$.
    \end{itemize}
    If $\word[1] \neq \word[2]$, then $\ell_{0} \dots \ell_{n-1} \in \hat{\val}(\word[1]) \setminus \hat{\val}(\word[2])$.
    Such a \kl{language valuation} $\val$ always exists.
\end{lemma}
\begin{proof}
    Since $\ell_{i} \in \val(x_i)$ for all $i$, we have $\ell_{0} \dots \ell_{n-1} \in \hat{\val}(\word[1])$.
    Assume that $\ell_{0} \dots \ell_{n-1} \in \hat{\val}(y_0 \dots y_{m-1})$.
    For $i = 0$, by the condition of $\hat{\val}(y_i)$ (i.e., $\ell_{i} \dots \ell_{j-1} \not\in \hat{\val}(y_i)$ unless $j = i + 1$),
    we should have $\ell_{i} \in \hat{\val}(y_i)$, $\ell_{i+1} \dots \ell_{n-1} \in \hat{\val}(y_{i+1} \dots y_{m-1})$, and $y_i = x_i$.
    By using the same argument iteratively, the condition above should hold for all $i < n$;
    thus, we have $\eps \in \hat{\val}(y_{n} \dots y_{m-1})$ and $y_0 \dots y_{n-1} = x_{0} \dots x_{m-1}$.
    Since $\eps \not\in \hat{\val}(y_{n})$, we have $y_{n} \dots y_{m-1} = \eps$; thus, $m = n$.
    However this yields $\word[1] = x_0 \dots x_{n-1} = y_{0} \dots y_{m-1} = \word[2]$, which contradicts the assumption.
    Hence $\ell_{0} \dots \ell_{n-1} \in \hat{\val}(\word[1]) \setminus \hat{\val}(\word[2])$.
    Additionally, such a \kl{language valuation} $\val$ always exists as follows.
    If some conditions conflict, then the first condition ($\ell_{i} \in \hat{\val}(x_i)$) and the second condition when $j = i+1$ ($\ell_{i} \in \hat{\val}(y_i)$) are for some $i$.
    If $y_i = x_i$, then $\ell_{i} \in \hat{\val}(y_i) = \hat{\val}(x_i)$, so it does not conflict to the condition $\ell_{i} \in \hat{\val}(x_i)$;
    If $y_i = x_i^{\compl}$ (or $y_i^{\compl} = x_i$), then $\ell_{i} \not\in \hat{\val}(y_i)$, so it does not conflict to the condition $\ell_{i} \in \hat{\val}(x_i)$.
    Otherwise, they are not conflicted, as the variables occurring in $x_i$ and $y_i$ are different.
    Thus, in either case, conditions are not conflicted.
    Hence, this completes the proof.
\end{proof}

\begin{theorem}[Completeness for words with variable complements]\label{theorem: completeness word}
    
    For all \kl{words} $\word[1], \word[2]$ over $\vsig'$,
    \[\LANG \models \word[1] = \word[2] \quad\iff\quad \word[1] = \word[2].\]
\end{theorem}
\begin{proof}
    $\Longleftarrow$:
    Clear.
    $\Longrightarrow$:
    The contraposition is shown by \Cref{lemma: completeness word}.
\end{proof}

\begin{remark}
    Since $\lang{\word[1]}_{\vsig'} = \set{\word[1]}$, \Cref{theorem: completeness word} also shows that: for all \kl{words} $\word[1], \word[2]$ over $\vsig'$,
    \[\lang{\word[1]}_{\vsig'} = \lang{\word[2]}_{\vsig'} \quad\iff\quad \LANG \models \word[1] = \word[2].\]
    However, for general terms, the direction $\Longleftarrow$ fails:
    For example, when $x \neq y$,
    \begin{align*}
        \LANG & \models x \union x^{\compl} =  y \union y^{\compl} & \lang{x \union x^{\compl}}_{\vsig'} \neq \lang{y \union y^{\compl}}_{\vsig'}.
    \end{align*}
    (The direction $\Longrightarrow$ always holds by \Cref{lemma: lang val}.)
    Thus, we need more axioms to characterize the equational theory.
\end{remark}

\begin{remark}
    As (\ref{equation: LANG to lang}) and \Cref{remark: LANG and lang}, for all \kl{words} $\word[1], \word[2]$ over $\vsig'$,
    \[\LANG \models \word[1] = \word[2] \quad\Longrightarrow\quad \lang{\word[1]} = \lang{\word[2]}.\]
    However, the converse direction fails even for \kl{words} (e.g., $\word[1] = x^{\compl}$ and $\word[2] = x^{\compl} x^{\compl}$).
\end{remark}

\subsection{Separating one-variable words with small number of letters}\label{section: completeness word one variable}
We write $\LANG_{n}$ for the class of \kl{language models} over a set of cardinality at most $n$.
We write
\[\LANG_{n} \models \term[1] = \term[2] \quad \defiff \quad \hat{\val}(\term[1]) = \hat{\val}(\term[2]) \mbox{ holds for all \kl{(language) valuations} $\val$ on \kl{$\sig$-algebras} in $\LANG_{n}$}.\]
Notice that \kl{words-to-letters valuations} need an unbounded number of \kl{letters}; so the proof of \Cref{theorem: completeness word} cannot directly apply to the class $\LANG_{n}$.
Nevertheless, for \emph{one-variable} \kl{words} (i.e., \kl{words} over the set $\set{z, z^{\compl}}$ where $z$ is a variable),
we can show completeness theorems (cf.\ \Cref{theorem: completeness word}) of the equational theory over $\LANG_{n}$, as \Cref{theorem: completeness word one variable unary,theorem: completeness word one variable}.
The valuation in the proof of \Cref{theorem: completeness word one variable} is given by the second author.%

For a \kl{word} $\word[1] = x_0 \dots x_{n-1} \in \set{z, z^{-}}^{*}$ and $x \in \set{z, z^{-}}$,
we write $\len{\word[1]}_{x}$ for the number $\#(\set{0 \le i < n \mid x_{i} = x})$.
For a \kl{letter} $a$ and $n \in \nat$, we write $a^{n}$ for the \kl{word} $a \dots a$ of \kl{length} $n$.
\begin{theorem}\label{theorem: completeness word one variable unary}
    For all \kl{words} $\word[1], \word[2] \in \set{z, z^{\compl}}^*$, we have:
    \[\LANG_{1} \models \word[1] = \word[2] \quad \iff \quad \len{\word[1]  }_{z} = \len{\word[2]}_{z} \;\land\; \len{\word[1]}_{z^{\compl}} = \len{\word[2]}_{z^{\compl}}.\]
\end{theorem}
\begin{proof}
    $\Longleftarrow$:
    By the following commutative law: for all \kl{language valuations} $\val$ over a set of cardinality at most $1$,
    $\hat{\val}(z z^{\compl}) = \hat{\val}(z^{\compl} z)$.
    $\Longrightarrow$:
    If $\len{\word[1]}_{z} < \len{\word[2]}_{z}$,
    then let $\val$ be the \kl{language valuation} defined by $\val(z) = \set{a}$.
    Then $a^{\len{\word[1]}_{z}} \in \hat{\val}(\word[1]) \setminus \hat{\val}(\word[2])$; thus $\LANG_{1} \not\models \word[1] = \word[2]$.
    If $\len{\word[1]}_{z^{\compl}} < \len{\word[2]}_{z^{\compl}}$,
    then let $\val$ be the \kl{language valuation} defined by $\val(z) = \set{a}^* \setminus \set{a}$.
    Then $a^{\len{\word[1]}_{z^{\compl}}} \in \hat{\val}(\word[1]) \setminus \hat{\val}(\word[2])$.
    If $\len{\word[1]}_{z} > \len{\word[2]}_{z}$ (resp.\ $\len{\word[1]}_{z^{\compl}} > \len{\word[2]}_{z^{\compl}}$), then similarly to the cases above.
\end{proof}

\begin{theorem}\label{theorem: completeness word one variable}
    For all \kl{words} $\word[1], \word[2] \in \set{z, z^{\compl}}^*$, we have:
    \[\LANG_{2} \models \word[1] = \word[2] \quad \iff \quad  \LANG \models \word[1] = \word[2] \quad\iff\quad \word[1] = \word[2].\]
\end{theorem}
\begin{proof}
    The two $\Longleftarrow$ are clear by definition.
    We prove $\word[1] \neq \word[2] \Longrightarrow \LANG_{2} \not\models \word[1] = \word[2]$.
    By \Cref{theorem: completeness word one variable unary},
    it suffices to show the case when $\len{\word[1]}_{z} = \len{\word[2]}_{z}$ and $\len{\word[1]}_{z^{\compl}} = \len{\word[2]}_{z^{\compl}}$.
    Let $n \defeq \len{\word[1]}_{z^{\compl}} = \len{\word[2]}_{z^{\compl}}$ and let $\word[1], \word[2]$ be as follows:
    \begin{align*}
        \word[1] & = z^{c_0} z^{\compl} z^{c_1} \dots z^{\compl} z^{c_n}  \\
        \word[2] & = z^{d_0} z^{\compl} z^{d_1} \dots z^{\compl} z^{d_n}.
    \end{align*}
    Since $\word[1] \neq \word[2]$,
    there is $i \le n$ such that $c_j = d_j$ for all $j < i$ and $c_i \neq d_i$.
    Without loss of generality, we can assume $c_i < d_i$.
    Now, we consider the following \kl{language valuation} $\val$ over $A \defeq \set{a, b}$:
    \begin{align*}
        \val(z) & \defeq \lang{(a A^*) \cap (A^* a)} = \set{c_0 \dots c_{n-1} \in \set{a, b}^* \mid n \ge 1, c_0 = a, c_{n-1} = a}.
    \end{align*}
    Then $a^{(\sum_{j = 0}^{i} c_j)} b a^{(\sum_{j = i+1}^{n} c_j)} \in \hat{\val}(\word[1])$, as $a \in \hat{\val}(z)$ and $\eps, b \in \hat{\val}(z^{\compl})$.
    Assume, towards contradiction, that $a^{(\sum_{j = 0}^{i} c_j)} b a^{(\sum_{j = i+1}^{n} c_j)} \in \hat{\val}(\word[2])$.
    Each $z$ occurring in $\word[2]$ should map to $a$, as $(\sum_{j=0}^{n} c_j) = (\sum_{j=0}^{n} d_j)$ and every word in $\val(z)$ except for $a$ has at least two occurrences of $a$.
    The $(\sum_{j=0}^{i} c_j)$-th occurrence and $((\sum_{j=0}^{i} c_j) + 1)$-th occurrence of $z^{-}$ are adjacent (since $(\sum_{j=0}^{i} c_j) < (\sum_{j=0}^{i} d_j)$).
    Combining them yields $b \in \hat{\val}(\eps)$, thus reaching a contradiction.
    Hence, $a^{(\sum_{j = 0}^{i} c_j)} b a^{(\sum_{j = i+1}^{n} c_j)} \in \hat{\val}(\word[1]) \setminus \hat{\val}(\word[2])$.
    This completes the proof.
\end{proof}
The proof of \Cref{theorem: completeness word one variable unary,theorem: completeness word one variable} only applies to one-variable words.
We leave open \Cref{theorem: completeness word one variable unary,theorem: completeness word one variable} for many variables words (cf.\ \Cref{theorem: completeness word}).
\section{Conclusion and future work}\label{section: conclusion}
We have introduced \kl{words-to-letters valuations}.
By using them, we have shown the decidability and complexity of the identity/variable/word inclusion problems (\Cref{corollary: identity decidable,corollary: variable decidable,corollary: word decidable}) and the universality problem (\Cref{corollary: universality decidable}) of the equational theory of KA terms with variable complements w.r.t.\ languages; in particular, the inequational theory $\term[1] \le \term[2]$ is coNP-complete when $\term[1]$ does not contain Kleene-star (\Cref{corollary: star-free}).
Additionally, we have proved a completeness theorem for words with variable complements w.r.t.\ languages (\Cref{theorem: completeness word});
moreover, for one-variable words, the equational theory over $\LANG$ coincides with that over $\LANG_{2}$ (\Cref{theorem: completeness word one variable}).

A natural interest is to extend our decidability results, e.g., for full KA terms with variable complements.
As \Cref{corollary: word witness} shows, even for full terms, \kl{words-to-letters valuations} are sufficient valuations in investigating the equational theory.
The first author %
conjectures that the equational theory of KA terms with variable complements is decidable, possibly by combining the technique like \emph{saturable paths} \cite{nakamuraExistentialCalculiRelations2023} (which were introduced for the equational theory w.r.t.\ binary relations).
Additionally, we leave open the (finite) axiomatizability of the equational theory (including that over sets of bounded cardinality; cf.\ \Cref{section: completeness word one variable}).
 \section*{Acknowledgement}
We would like to thank the anonymous reviewers for their useful comments.
This work was supported by JSPS KAKENHI Grant Number JP21K13828 and JST ACT-X Grant Number JPMJAX210B, Japan.
\bibliographystyle{eptcs}
\bibliography{main}

\begin{thebibliography}{10}
\providecommand{\bibitemdeclare}[2]{}
\providecommand{\surnamestart}{}
\providecommand{\surnameend}{}
\providecommand{\urlprefix}{Available at }
\providecommand{\url}[1]{\texttt{#1}}
\providecommand{\href}[2]{\texttt{#2}}
\providecommand{\urlalt}[2]{\href{#1}{#2}}
\providecommand{\doi}[1]{doi:\urlalt{https://doi.org/#1}{#1}}
\providecommand{\eprint}[1]{arXiv:\urlalt{https://arxiv.org/abs/#1}{#1}}
\providecommand{\bibinfo}[2]{#2}

\bibitemdeclare{article}{andrekaEquationalTheoryKleene2011}
\bibitem{andrekaEquationalTheoryKleene2011}
\bibinfo{author}{Hajnal \surnamestart Andréka\surnameend},
  \bibinfo{author}{Szabolcs \surnamestart Mikulás\surnameend} \&
  \bibinfo{author}{István \surnamestart Németi\surnameend}
  (\bibinfo{year}{2011}): \emph{\bibinfo{title}{The equational theory of Kleene
  lattices}}.
\newblock {\slshape \bibinfo{journal}{Theoretical Computer Science}}
  \bibinfo{volume}{412}(\bibinfo{number}{52}), pp. \bibinfo{pages}{7099--7108},
  \doi{10.1016/J.TCS.2011.09.024}.

\bibitemdeclare{article}{bloomNotesEquationalTheories1995}
\bibitem{bloomNotesEquationalTheories1995}
\bibinfo{author}{S.~L. \surnamestart Bloom\surnameend},
  \bibinfo{author}{Z.~\surnamestart Ésik\surnameend} \& \bibinfo{author}{Gh.
  \surnamestart Stefanescu\surnameend} (\bibinfo{year}{1995}):
  \emph{\bibinfo{title}{Notes on equational theories of relations}}.
\newblock {\slshape \bibinfo{journal}{algebra universalis}}
  \bibinfo{volume}{33}(\bibinfo{number}{1}), pp. \bibinfo{pages}{98--126},
  \doi{10.1007/BF01190768}.

\bibitemdeclare{book}{Conway1971}
\bibitem{Conway1971}
\bibinfo{author}{John~H. \surnamestart Conway\surnameend}
  (\bibinfo{year}{1971}): \emph{\bibinfo{title}{Regular Algebra and Finite
  Machines}}.
\newblock \bibinfo{publisher}{Chapman and Hall}.

\bibitemdeclare{inproceedings}{cookComplexityTheoremprovingProcedures1971}
\bibitem{cookComplexityTheoremprovingProcedures1971}
\bibinfo{author}{Stephen~A. \surnamestart Cook\surnameend}
  (\bibinfo{year}{1971}): \emph{\bibinfo{title}{The complexity of
  theorem-proving procedures}}.
\newblock In: {\slshape \bibinfo{booktitle}{STOC}}, \bibinfo{publisher}{ACM},
  p. \bibinfo{pages}{151–158}, \doi{10.1145/800157.805047}.

\bibitemdeclare{article}{Hunt1976}
\bibitem{Hunt1976}
\bibinfo{author}{Harry~B. \surnamestart Hunt~III\surnameend},
  \bibinfo{author}{Daniel~J. \surnamestart Rosenkrantz\surnameend} \&
  \bibinfo{author}{Thomas~G. \surnamestart Szymanski\surnameend}
  (\bibinfo{year}{1976}): \emph{\bibinfo{title}{On the equivalence,
  containment, and covering problems for the regular and context-free
  languages}}.
\newblock {\slshape \bibinfo{journal}{Journal of Computer and System Sciences}}
  \bibinfo{volume}{12}(\bibinfo{number}{2}), pp. \bibinfo{pages}{222--268},
  \doi{10.1016/S0022-0000(76)80038-4}.

\bibitemdeclare{techreport}{kleeneRepresentationEventsNerve1951}
\bibitem{kleeneRepresentationEventsNerve1951}
\bibinfo{author}{S.~C. \surnamestart Kleene\surnameend} (\bibinfo{year}{1951}):
  \emph{\bibinfo{title}{Representation of Events in Nerve Nets and Finite
  Automata}}.
\newblock \bibinfo{type}{Technical Report}, \bibinfo{institution}{RAND
  Corporation}.
\newblock
  \urlprefix\url{https://www.rand.org/pubs/research_memoranda/RM704.html}.

\bibitemdeclare{inproceedings}{kozenKleeneAlgebraTests1996}
\bibitem{kozenKleeneAlgebraTests1996}
\bibinfo{author}{Dexter \surnamestart Kozen\surnameend} \&
  \bibinfo{author}{Frederick \surnamestart Smith\surnameend}
  (\bibinfo{year}{1996}): \emph{\bibinfo{title}{{Kleene} algebra with tests:
  Completeness and decidability}}.
\newblock In: {\slshape \bibinfo{booktitle}{CSL}}, {\slshape
  \bibinfo{series}{LNCS}} \bibinfo{volume}{1258},
  \bibinfo{publisher}{Springer}, pp. \bibinfo{pages}{244--259},
  \doi{10.1007/3-540-63172-0_43}.

\bibitemdeclare{inproceedings}{meyerEquivalenceProblemRegular1972}
\bibitem{meyerEquivalenceProblemRegular1972}
\bibinfo{author}{A.~R. \surnamestart Meyer\surnameend} \&
  \bibinfo{author}{L.~J. \surnamestart Stockmeyer\surnameend}
  (\bibinfo{year}{1972}): \emph{\bibinfo{title}{The equivalence problem for
  regular expressions with squaring requires exponential space}}.
\newblock In: {\slshape \bibinfo{booktitle}{SWAT}}, \bibinfo{publisher}{IEEE},
  pp. \bibinfo{pages}{125--129}, \doi{10.1109/SWAT.1972.29}.

\bibitemdeclare{inproceedings}{nakamuraExistentialCalculiRelations2023}
\bibitem{nakamuraExistentialCalculiRelations2023}
\bibinfo{author}{Yoshiki \surnamestart Nakamura\surnameend}
  (\bibinfo{year}{2023}): \emph{\bibinfo{title}{Existential Calculi of
  Relations with Transitive Closure: Complexity and Edge Saturations}}.
\newblock In: {\slshape \bibinfo{booktitle}{LICS}}, \bibinfo{publisher}{IEEE},
  pp. \bibinfo{pages}{1--13}, \doi{10.1109/LICS56636.2023.10175811}.

\bibitemdeclare{phdthesis}{ngRelationAlgebrasTransitive1984}
\bibitem{ngRelationAlgebrasTransitive1984}
\bibinfo{author}{Kan~Ching \surnamestart Ng\surnameend} (\bibinfo{year}{1984}):
  \emph{\bibinfo{title}{Relation algebras with transitive closure}}.
\newblock Ph.D. thesis, \bibinfo{school}{University of California}.

\bibitemdeclare{inproceedings}{pousCompletenessTheoremsKleene2022}
\bibitem{pousCompletenessTheoremsKleene2022}
\bibinfo{author}{Damien \surnamestart Pous\surnameend} \& \bibinfo{author}{Jana
  \surnamestart Wagemaker\surnameend} (\bibinfo{year}{2022}):
  \emph{\bibinfo{title}{Completeness Theorems for Kleene Algebra with Top}}.
\newblock In: {\slshape \bibinfo{booktitle}{CONCUR}}, \bibinfo{volume}{243},
  \bibinfo{publisher}{Schloss Dagstuhl}, pp. \bibinfo{pages}{26:1--26:18},
  \doi{10.4230/LIPICS.CONCUR.2022.26}.

\bibitemdeclare{inproceedings}{Meyer1973}
\bibitem{Meyer1973}
\bibinfo{author}{Larry~J. \surnamestart Stockmeyer\surnameend} \&
  \bibinfo{author}{Albert~R. \surnamestart Meyer\surnameend}
  (\bibinfo{year}{1973}): \emph{\bibinfo{title}{Word problems requiring
  exponential time (Preliminary Report)}}.
\newblock In: {\slshape \bibinfo{booktitle}{STOC}}, \bibinfo{publisher}{ACM},
  pp. \bibinfo{pages}{1--9}, \doi{10.1145/800125.804029}.

\bibitemdeclare{article}{Tarski1941}
\bibitem{Tarski1941}
\bibinfo{author}{Alfred \surnamestart Tarski\surnameend}
  (\bibinfo{year}{1941}): \emph{\bibinfo{title}{On the Calculus of Relations}}.
\newblock {\slshape \bibinfo{journal}{The Journal of Symbolic Logic}}
  \bibinfo{volume}{6}(\bibinfo{number}{3}), pp. \bibinfo{pages}{73--89},
  \doi{10.2307/2268577}.

\bibitemdeclare{article}{thompsonProgrammingTechniquesRegular1968}
\bibitem{thompsonProgrammingTechniquesRegular1968}
\bibinfo{author}{Ken \surnamestart Thompson\surnameend} (\bibinfo{year}{1968}):
  \emph{\bibinfo{title}{Programming Techniques: Regular expression search
  algorithm}}.
\newblock {\slshape \bibinfo{journal}{Communications of the ACM}}
  \bibinfo{volume}{11}(\bibinfo{number}{6}), pp. \bibinfo{pages}{419--422},
  \doi{10.1145/363347.363387}.

\end{thebibliography}

\appendix
\section{A direct proof of the coincidence between the equational theory w.r.t.\ languages and the language equivalence for KA terms} \label{section: LANG and lang}
(In this section, we use the notations of \Cref{section: preliminaries}.)

We say that a term $\term$ is a \intro*\kl{KA term} if the complement ($\bl^{\compl}$) does not occur in $\term$.
Recall \kl{language valuations} for \kl{languages} in \Cref{section: term to lang}.
\begin{lemma}[cf.\ \Cref{lemma: lang val}]\label{lemma: lang val KA}
    Let $\val$ be a \kl{language valuation}.
    For all \kl{KA terms} $\term$, we have
    \[\hat{\val}(\term) \quad=\quad \hat{\val}(\lang{\term}).\]
\end{lemma}
\begin{proof}
    By \Cref{lemma: lang val}, as $\lang{\term} = \lang{\term}_{\vsig'}$ (since \kl{KA terms} do not contain the complement).
\end{proof}

\begin{theorem}\label{theorem: LANG and lang}
    For all \kl{KA terms} $\term[1], \term[2]$,
    \[\LANG \models \term[1] = \term[2] \quad \iff \quad \lang{\term[1]} = \lang{\term[2]}.\]
\end{theorem}
\begin{proof}
    We have
    \begin{align*}
        \LANG \models \term[1] = \term[2] & \quad\Longrightarrow\quad \lang{\term[1]} = \lang{\term[2]} \tag{$\lang{\bl}$ is an instance of \kl{language valuations}}                                 \\
                                          & \quad\Longrightarrow\quad \mbox{for all \kl{language valuations} $\val$}, \hat{\val}(\lang{\term[1]}) = \hat{\val}(\lang{\term[2]})                       \\
                                          & \quad\Longleftrightarrow\quad \mbox{for all \kl{language valuations} $\val$}, \hat{\val}(\term[1]) = \hat{\val}(\term[2]) \tag{\Cref{lemma: lang val KA}} \\
                                          & \quad\Longleftrightarrow\quad  \LANG \models \term[1] = \term[2]. \tag{By definition}
    \end{align*}
\end{proof}

\end{document}